\documentclass[
 jmp,
 amsmath,amssymb,
 reprint,%
author-numerical,%
]{revtex4-1}

\usepackage{graphicx}
\usepackage{dcolumn}
\usepackage{bm}
\usepackage{tikz}
\usepackage{stmaryrd}

\usepackage[utf8]{inputenc}
\usepackage[T1]{fontenc}
\usepackage{mathptmx}

\usepackage{amsthm}
\newtheorem{dfn}{Definition} 
\newtheorem{thm}{Theorem}
\newtheorem{prop}{Proposition}

\newtheorem{ex}{Example}

\begin{document}

\preprint{AIP/123-QED}

\title[Sample title]{Duality Between Quantization and Measurement}

\author{Yosuke Morimoto} 
 \email{myosuke@post.kek.jp / hrmorimotohm@gmail.com}
\affiliation{Department of Physics, Graduate School of Science, the University of Tokyo, 7-3-1, Hongo, Bunkyo-ku, Tokyo 113-0033, Japan.}

\date{\today}

\begin{abstract}
Using the known duality between events and states, we establish the fact that there is a duality between quantization of events and measurement of states.
With this duality, we will show a generalized system of imprimitivity and a covariant measurement are in a dual relation.
We also propose a novel quantization scheme as a dual of a measurement model.
\end{abstract}

\maketitle

\section{Introduction}
The structure of events in quantum theory has been studied for a long time.
In \cite{DFMB}, the authors proposed an effect algebras to formulate the structure of events in a context of quantum logic.
An effect algebra is a set of events with operations such as the union of events and the negation. 
An effect algebra was generalized to an effect module in \cite{SGSP} admitting a probabilistic discard of events.

In \cite{BJ}, the author showed that a category of effect modules $\mathbf{EMod}$  is an dual adjunction to a category of convex spaces $\mathbf{Conv}$.
Convex space is a natural structure of states, so they established an adjunction between events and states.
In \cite{BJJM1} this adjunction is generalized to a categorical duality by imposing some topologies.
They established the important fact that events and states are in a dual relation.
In other words, events are completely characterized by states and vice versa.

There are several example of (Banach) effect modules.
We especially focus on [0,1]-valued continuous functions on a compact Hausdorff space $C(X,[0,1])$ and effect operators on a Hilbert space $Ef(\mathcal{H})$ since the former represents the structure of classical events and the latter represents the structure of quantum events. 
As the dual of these, we also focus on probability Radon measures on a compact Hausdorff space $PR(X)$ and density operators $DM(\mathcal{H})$ on a Hilbert space where the former represents the structure of classical states and the latter represents quantum states.

Incidentally, we can consider a map $Q:C(X,[0,1])\to Ef(\mathcal{H})$ which indicates quantization of events.
Also, we can consider a map $M:DM(\mathcal{H})\to PR(X)$ which indicates measurement of states.
Since $C(X,[0,1])$ is dual to $PR(X)$ and $Ef(\mathcal{H})$ is dual to $DM(\mathcal{H})$,  quantization of events $Q$ and measurement of states $M$ should be in a dual relation.
We show this fact in this paper.
(We make an additional remark that an adjoint relation between quantization and quasi-classicalization is discussed in \cite{JLIT}.)

We can apply the duality between quantization of events and measurement of states to an actual quantization and an actual measurement.
For an actual quantization, we will show that a dual of a generalized system of imprimitivity is a covariant measurement.
For an actual measurement, we will propose a novel quantization scheme as a dual of a measurement model.

\section{duality between Event and State}
In this section, we will review a duality between a category of Banach effect module $\mathbf{BEM}$ and  a category of compact convex Hausdorff space with a separating property $\mathbf{CCH}$\cite{BJJM1,BJJM2,BJ}.
This is equivalent to the physical statement that events and states are in a dual relation.

We introduce a set of continuous functions from a compact Hausdorff space to the interval $[0,1]$ with a supremum norm topology and a set of effect operators on a Hilbert space with an operator norm as examples of $\mathbf{BEM}$. 
We also introduce a set of a probability Radon measure on a compact Hausdorff space with a locally convex topology and a set of density operators on a Hilbert space with a locally convex topology as examples of $\mathbf{CCH}$.
We will give a brief description of these since these play an important role in this paper. 

\subsection{Dual adjunction between $\mathbf{EMod }$ and $\mathbf{Conv}$}
In this section, we see the fact that there is a dual adjunction between $\mathbf{EMod }$ and $\mathbf{Conv}$ by Hom functors \cite{BJ}.

To this end, we state the definition of effect modules and their homomorphisms and introduce a category of effect modules $\mathbf{EMod}$.
\begin{dfn}
An effect module consists of a set E with a partial operation $\ovee :E\times E\to E$ which is both commutative and associative, a neutral element 0, an orthosupplement $\neg : E\to E$ which has the property that $\neg x$ is the unique element with $x\ovee\neg x=1$ where $1=\neg0$, and $[0,1]$-action $\cdot : [0,1] \times E\to E$ satisfying following conditions:
\begin{gather*}
x\ovee 0=0\ovee x=x\\
x\bot 1 \mspace{5mu} \text{only if} \mspace{5mu}x=0\\
1\cdot x=x\\
r+s\le 1 \Rightarrow (r+s)\cdot x=r\cdot x+s\cdot x\\
(rs)\cdot x=r\cdot (s\cdot x)\\
x\bot y \Rightarrow r\cdot(x\ovee y)=(r\cdot x)\ovee (r\cdot y)
\end{gather*}
where $r,s\in [0,1]$ and $x,y\in E$ and one writes $x\bot y$ if the operation $x\ovee y$ is defined.

A homomorphism of effect modules $(E_1,0_{1},\neg_{1},\ovee_{1},\cdot_{1})\to(E_2,0_{2},\neg_{2},\ovee_{2},\cdot_{2})$ is a function $f : E_1\to E_2$ between the underlying sets satisfying
\begin{gather*}
f(r\cdot_{1}x)=r\cdot_{2}f(x)\\
f(1_{1})=1_{2}\\
x\bot y \Rightarrow f(x\ovee_{1} y)=f(x)\ovee_{2}f(y)
\end{gather*}
\end{dfn}
This yields a category of effect modules $\mathbf{EMod}$ whose objects consist of effect modules and arrows consist of their homomorphisms.
Note that an effect module describes the structure of events.

Next, we state the definition of convex spaces and their homomorphisms and introduce a category of  convex spaces $\mathbf{Conv}$.
\begin{dfn}
A convex space consists of a set S with a ternary operation $c_{-}(-,-) : [0,1]\times S\times S\to S$ satisfying following conditions:
\begin{gather*}
c_r(x,x)=x\\
c_r(x,y)=c_{1-r}(y,x)\\
c_0(x,y)=y\\
c_r(x,c_s(y,z))=c_{r+(1-r)s}(c_{r/[r+(1-r)s]}(x,y),z)
\end{gather*}
where $r,s\in [0,1]$, $x,y,z\in S$ and $r/[r+(1-r)s]\neq0$.

A homomorphism of convex sets $(S_1,c_{-}(-,-)_1)\to(S_2,c_{-}(-,-)_2)$ is a function $f : S_1\to S_2$ between the underlying sets satisfying
\begin{equation*}
f(c_r(x,x')_1)=c_r(f(x),f(x'))_2
\end{equation*}
for all $r\in [0,1]$ and $x,x'\in S_1$.
\end{dfn}
This yields a category of convex spaces $\mathbf{Conv}$ whose objects consist of convex sets and arrows consist of their homomorphisms.
Note that a convex space dscribes the structure of states.

We can construct a functor $\mathrm{Hom}_{\mathbf{Conv}}(-,[0,1]) : \mathbf{Conv}^{\mathrm{op}}\to \mathbf{EMod}$
where $\mathrm{Hom}_{\mathbf{Conv}}(S,[0,1])$ has an effect module structure for each $S\in\mathbf{Conv}$ with the definition of following :
$f\bot g \Leftrightarrow f(x)+g(x)\le 1 \mspace{5mu} \text{for all} \mspace{5mu}x\in S$, 
the pointwise sum $(f\ovee g)(x)=f(x)+g(x)$,
the pointwise orthosupplement $(\neg f)(x)=1-f(x)$
and the pointwise multiplication $(r\cdot f)(x)=r(f(x))$
for all $f,g\in\mathrm{Hom}_{\mathbf{Conv}}(S,[0,1])$.

We can also construct a functor $\mathrm{Hom}_{\mathbf{EMod}}(-,[0,1]) : \mathbf{Emod}^{\mathrm{op}}\to\mathbf{Conv}$ where $\mathrm{Hom}_{\mathbf{EMod}}(E,[0,1])$ has a convex space structure for each $E\in\mathbf{EMod}$ with the definition of following : the pointwise convex operation  $c_{p}(f,g)=pf+(1-p)g$ for all $f,g\in\mathrm{Hom}_{\mathbf{EMod}}(E,[0,1])$.

Under this setting, we can state a dual adjunction between $\mathbf{EMod }$ and $\mathbf{Conv}$.
\begin{thm}
functors $\mathrm{Hom}_{\mathbf{Conv}}(-,[0,1]) : \mathbf{Conv}^{\mathrm{op}}\to \mathbf{EMod}$ and $\mathrm{Hom}_{\mathbf{EMod}}(-,[0,1]) : \mathbf{Emod}^{\mathrm{op}}\to\mathbf{Conv}$ give a dual adjunction between a category of effect modules $\mathbf{EMod }$ and a category of convex spaces $\mathbf{Conv}$.
\end{thm}

\subsection{Duality between $\mathbf{BEM }$ and $\mathbf{CCH}$}
In this section, we see the fact that  the dual adjunction between $\mathbf{EMod }$ and $\mathbf{Conv}$ restricts to a duality between $\mathbf{BEM}$ and $\mathbf{CCH}$ \cite{BJJM1}.

To this end, we state the definition  of  Banach effect modules and introduce a category of Banach effect modules $\mathbf{BEM}$.
\begin{dfn}
A Banach effect module is an effect module which is complete with the metric given by following:
\begin{eqnarray*}
d(x,y)=\mathrm{max}(\mathrm{inf}\{r\in [0,1] | \frac{1}{2}x\le \frac{1}{2}y\ovee\frac{r}{2}1\},\\
\mathrm{inf}\{r\in [0,1] | \frac{1}{2}y\le \frac{1}{2}x\ovee\frac{r}{2}1\})
\end{eqnarray*}
where $x$ and $y$ are elements of an effect module.
A partial order on an effect module is given by $x\le y$ if and only if $x\ovee z=y$ for some element $z$.  
\end{dfn}
This yields category Banach effect modules $\mathbf{BEM}$ whose objects consists of Banach effect Modules and arrows consist of continuous effect module homomorphisms.

Next, we state the definition of convex compact Hausdorff spaces.
\begin{dfn}
Convex compact Hausdorff spaces is convex spaces with a compact Haudorff topology.
\end{dfn}
Before proceeding to a category of convex compact Hausdorff spaces, we mention a separating property.
The separating property on a convex compact Hausdorff space $S$ means that if $x\neq y$ for $x,y\in S$, there is a continuous convex homomorphism $s: S\to[0,1]$ such that $s(x)\neq s(y)$.
With the separating property, a convex compact Hausdorff space can be regarded as a subspace of a locally convex topological vector space.
The separating property is assumed in this paper.

A category of convex compact Hausdorff spaces $\mathbf{CCH}$ consists of convex compact Hausdorff spaces as objects and continuous convex homomorphisms.

 Now, we can state the categorical duality.
 Note that this is a duality between events and states.
 \begin{thm}
functors $\mathrm{Hom}_{\mathbf{CCH}}(-,[0,1]) : \mathbf{CCH}^{\mathrm{op}}\to \mathbf{BEM}$ and $\mathrm{Hom}_{\mathbf{BEM}}(-,[0,1]) : \mathbf{BEM}^{\mathrm{op}}\to\mathbf{CCH}$ give a duality between categories $\mathbf{BEM}$ and $\mathbf{CCH}$, where $\mathrm{Hom}_{\mathbf{CCH}}(S,[0,1])$ has an uniform topology and $\mathrm{Hom}_{\mathbf{BEM}}(E,[0,1])$ has a locally convex topology.
\end{thm}

\subsection{Examples and their consequences}
In this subsection, we present some examples of (Banach) effect module and convex (compact Hausdorff) space.
They are significant objects in this paper.
By applying a dual adjunction between $\mathbf{EMod}$ and $\mathbf{Conv}$ and a duality between $\mathbf{BEM}$ and $\mathbf{CCH}$ to these examples,  we can obtain an useful result.

First, we will present two examples of (Banach) effect modules.
\begin{ex}
A set of $[0,1]$-valued continuous functions on a compact Hausdorff space $C(X,[0,1])$ is an example of effect modules with the interpretation that $f_1\ovee f_2=f_1+f_2$ if $f_1+f_2\le 1$, a neutral element is a zero function and $\neg f=1-f$ for $f,f_1,f_2\in C(X,[0,1])$.
If a supremum norm topology is imposed on $C(X,[0,1])$, it becomes a Banach effect module. 
\end{ex} 
\begin{ex}
A set of effect operators on a Hilbert space $Ef(\mathcal{H})=\{E\in B(\mathcal{H}) | E^*=E, 0\le E\le I\}$
is an example of effect modules with the interpretation that $E_1\ovee E_2=E_1+E_2$ if $E_1+E_2\le I$, a neutral element is a zero operator and $\neg E=I-E$ for $E,E_1,E_2\in Ef(\mathcal{H})$.
If an operator norm topology is imposed on $Ef(\mathcal{H})$, it becomes a Banach effect module.
\end{ex} 
Note that $C(X,[0,1])$ corresponds to classical events and $Ef(\mathcal{H})$ corresponds to quantum events.

Next, we will present two examples of convex (compact Hausdorff) spaces.
\begin{ex}
A set of probability Radon measures on a compact Hausdorff space $PR(X)$ is an example of convex spaces with the convex operation $c_r(\mu_1,\mu_2)=r\mu_1+(1-r)\mu_2$ for $\mu_1,\mu_2\in PR(X)$.
If a locally convex topology is induced by a family of semi-norms $\{|\int_X fd(\cdot)| | f\in C(X)\}$, $PR(X)$ becomes  a convex compact Hausdorff space.
\end{ex}

\begin{ex}
A set of density operators on a Hilbert space $DM(\mathcal{H})=\{\rho\in T(\mathcal{H}) | \rho^*=\rho, Tr(\rho)=1\}$ is an example of convex spaces with the convex operation $c_r(\rho_1,\rho_2)=r\rho_1+(1-r)\rho_2$ for $\rho_1,\rho_2\in DM(\mathcal{H})$.
If a locally convex topology is induced by a family of semi-norms $\{|Tr[(\cdot)B]| | E\in B(\mathcal{H})\}$, $DM(\mathcal{H})$  becomes a convex compact Hausdorff space.
\end{ex}

Note that $PR(X)$ corresponds to classical states and $DM(\mathcal{H})$ corresponds to quantum states.

By applying the duality to these examples, we can obtain an useful result.
In following discussion, we assume topologies in examples above.
\begin{prop}
Let $\Phi$ a continuous affine functional on $PR(X)$.
Then there exist an unique $f\in C(X,[0,1])$ such that $\Phi(\mu)=\int_X fd\mu$  for every $\mu\in PR(X)$
\end{prop}
\begin{proof}
Let $g\in\mathrm{Hom}_{\mathbf{Conv}}(PR(X),\mathrm{Hom}_{\mathbf{Emod}}(C(X,[0,1]),[0,1]))$ and $h\in\mathrm{Hom}_{\mathbf{Emod}}(C(X,[0,1]),\mathrm{Hom}_{\mathbf{Conv}}(PR(X),[0,1])$.
Applying the Riesz representation theorem, $g$ can be written in a form $\int_X(-)d(-)$.
By the dual adjunction, there is a bijective correspondence between $g$ and $h$.
Combining these,  $h$ can be written in the form $\int_X(-)d(-)$.
By the duality, there is a bijective correspondence between $C(X,[0,1])$ and $\mathrm{Hom}_{\mathbf{CCH}}(PR(X),[0,1])$ so $\Phi\in\mathrm{Hom}_{\mathbf{CCH}}(PR(X),[0,1])$ can be written in an integral form $\int_Xfd(-)$ using $f\in C(X,[0,1])$.
\end{proof}

\section{duality between Quantization and Measurement }
In this section, we introduce a notion of quantization of events and measurement of states.
We establish the fact that there is a duality between quantization of events and measurement of states.
After that we will apply that duality to concrete examples of quantization and measurement.

\subsection{Duality between Quantization of Events and Measurement of States}
In this subsection, we define what quantization of events and measurement of states mean.
We also investigate a relation between quantization of events and measurement of states for a later discussion. 
\begin{dfn}
We call a Banach Effect module homomorphism $Q:C(X,[0,1])\to Ef(\mathcal{H})$ quantization of events from compact Hausdorff space $X$ to Hilbert space $\mathcal{H}$.
\end{dfn}
\begin{dfn}
We call a continuous convex homomorphism $M:DM(\mathcal{H})\to PR(X)$ measurement of states from Hilbert space $\mathcal{H}$ to compact Hausdorff space $X$.
\end{dfn}
Quantization of events is a map from classical events to quantum events.
Measurement of states is a map from quantum states to classical states.
Noting the fact that $C(X,[0,1])$ is dual to $PR(X)$ and $Ef(\mathcal{H})$ is dual to $DM(\mathcal{H})$ by the $\mathbf{BEM}$-$\mathbf{CCH}$ duality, we may notice quantization of events is dual to measurements of states.

First, let us induce measurement of states from a dual of quantization of events. 
Since measurement of states can be regarded as a map $(Ef(\mathcal{H})\to [0,1])\to (C(X,[0,1])\to [0,1])$, we require following diagram commute:\\
\begin{tikzpicture}[auto]
\node (a) at (0, 1.5) {$C(X,[0,1])$}; \node (x) at (3, 1.5) {$[0,1]$};
\node (b) at (0, 0) {$Ef(\mathcal{H})$};   \node (y) at (3, 0) {$[0,1]$};
\draw[->] (a) to node {$\alpha$} (x);
\draw[->] (x) to node {$id$} (y);
\draw[->] (a) to node[swap] {$Q$} (b);
\draw[->] (b) to node[swap] {$\beta$} (y);
\end{tikzpicture}
\\In other words, for maps $\alpha :C(X,[0,1])\to [0,1]$, $\beta :Ef(\mathcal{H})\to [0,1]$ and quantization of events $Q$,
\begin{equation*}
\alpha(f)=\beta\circ Q(f)
\end{equation*}
 holds for all $f\in C(X,[0,1])$. 
Since quantization of events $Q$ is an unital positive linear map from $C(X,[0,1])$ to $Ef(\mathcal{H})$, there exists an unique positive operator valued measure (POVM) $E$ such that 
\begin{equation*}
Q(f)=\int_X f dE
\end{equation*}
for all $f\in C(X,[0,1])$.
In addition, by the Bucsh theorem \cite{PB} there exists an unique density operator $\rho$ such that 
\begin{equation*}
\beta (A)=\mathrm{Tr}[\rho A]
\end{equation*}
 for all $A\in Ef(\mathcal{H})$.
 By equations above and the Rieaz representation theorem, an unique measure for $\alpha$ is $\mathrm{Tr}[\rho E]$.
 So a dual of quantization of events as a map $\beta\to\alpha$ induces measurement of states $M: \rho\in DM(\mathcal{H})\to \mathrm{Tr}[\rho E]\in PR(X)$ which is characterized by the  POVM appeared in quantization of events.
 Summing up these results, we can state as following:
 \begin{prop}
 Quantisztion  of events $Q:C(X,[0,1])\to Ef(\mathcal{H})$ can be written as $Q=\int_X (-)dE$ by using a POVM $E$ and induces measurement of states $M:DM(\mathcal{H})\to PR(X)$ which can be written as $M=\mathrm{Tr}[(-)E]$.
 This fact can be represented in following diagram:\\
\begin{tikzpicture}[auto]
\node (a) at (0, 1.5) {$C(X,[0,1])$}; \node (x) at (3, 1.5) {$[0,1]$};
\node (b) at (0, 0) {$Ef(\mathcal{H})$};   \node (y) at (3, 0) {$[0,1]$};
\node (s) at (4.5,0) {$DM(\mathcal{H})$};   \node (t) at (4.5,1.5) {$PR(X)$};
\draw[->] (a) to node {$\int_X(-)d(M\rho)$} (x);
\draw[->] (x) to node {$id$} (y);
\draw[->] (a) to node[swap] {$Q=\int_X (-)dE$} (b);
\draw[->] (b) to node[swap] {$\mathrm{Tr}[\rho (-)]$} (y);
\draw[->] (s) to node[swap] {$M=\mathrm{Tr}[(-)E]$} (t);
\end{tikzpicture}
\end{prop}

Second, let us induce quantization of events from a dual of measurement of states.
Since quantization of events can be regarded as a map $(DM(\mathcal{H})\to [0,1])\to (PR(X)\to [0,1])$, we require following diagram commute:\\
\begin{tikzpicture}[auto]
\node (a) at (0, 1.5) {$DM(\mathcal{H})$}; \node (x) at (3, 1.5) {$[0,1]$};
\node (b) at (0, 0) {$PR(X)$};   \node (y) at (3, 0) {$[0,1]$};
\draw[->] (a) to node {$\gamma$} (x);
\draw[->] (x) to node {$id$} (y);
\draw[->] (a) to node[swap] {$M$} (b);
\draw[->] (b) to node[swap] {$\delta$} (y);
\end{tikzpicture}
\\In other words, for maps $\gamma :DM(\mathcal{H})\to [0,1]$, $\beta :PR(X)\to [0,1]$ and measurement of states $M$,
\begin{equation*}
\gamma(\rho)=\delta\circ M(\rho)
\end{equation*}
 holds for all $\rho\in DM(\mathcal{H})$. 
Since a dual of trace class operator on a Hilbert space is bounded operator on the Hilbert space, $T(\mathcal{H})^*\cong B(\mathcal{H})$, there exists an unique POVM $E$ such that 
\begin{equation*}
M(\rho)=\mathrm{Tr}[\rho E]
\end{equation*}
 for all $\rho\in DM(\mathcal{H})$.
In addition, by the proposition 1 there exists an unique $f\in C(X,[0,1])$ such that 
\begin{equation*}
\delta(\mu)=\int_X fd\mu
\end{equation*}
for all $\mu\in PR(X)$.
By quations above and $T(\mathcal{H})^*\cong B(\mathcal{H})$, an unique effect for $\gamma$ is $\int_x fdE$.
So a dual of measurement of states as a map $\gamma\to\delta$ induces quantization of events $Q: f\in C(X,[0,1]) \to \int_x fdE\in Ef(\mathcal{H})$ which is characterized by the POVM appeared in measurement of states.
 Summing up these results, we can state as following:
 \begin{prop}
 Measurement of states $M:DM(\mathcal{H})\to PR(X)$ can be written as  $M=\mathrm{Tr}[(-)E]$ by using a POVM $E$ and induces quantization of events $Q:C(X,[0,1])\to Ef(\mathcal{H})$ which can be written as $Q=\int_X (-)dE$.
 This fact can be represented in following diagram.\\
\begin{tikzpicture}[auto]
\node (a) at (0, 1.5) {$DM(\mathcal{H})$}; \node (x) at (3, 1.5) {$[0,1]$};
\node (b) at (0, 0) {$PR(X)$};   \node (y) at (3, 0) {$[0,1]$};
\node (s) at (4.5,0) {$C(X,[0,1])$};   \node (t) at (4.5,1.5) {$Ef(\mathcal{H})$};
\draw[->] (a) to node {$\mathrm{Tr}[(-)Q(f)]$} (x);
\draw[->] (x) to node {$id$} (y);
\draw[->] (a) to node[swap] {$M=\mathrm{Tr}[(-)E]$}(b);
\draw[->] (b) to node[swap] {$\int_X fd(-)$} (y);
\draw[->] (s) to node[swap]  {$Q=\int_X (-)dE$} (t);
\end{tikzpicture}
\end{prop}
From the proposition 2 and 3, we can see that quantization of events and measurement of states are in a dual relation characterized by a single POVM.

\subsection{Some Examples of Duality}

In this subsection, we will see how the duality works in concrete examples of quantization and measurements.

First example is a generalized system of imprimitivity.
This is a system of imprimitivity \cite{GM} employing a POVM instead of a projection valued measure.
This provides canonical quantization on a homogeneous space.
For a detailed explanation, see \cite{KL}.

Let $G$ be a topological group and $X$ be a $G$-space.
A generalized system of imprimitivity on $G$ is a quadruple $(\mathcal{H},U,X,P)$ where $U$ is an unitary representation of $G$ on a Hilbert apace $\mathcal{H}$ and POVM $E$ on $X$ such that
\begin{equation*}
U(g)E(\Delta)U(g)^{-1}=E(g\Delta) 
\end{equation*}  
holds for all $g\in G$ and all Borel sets $\Delta\subset X$.
This condition is equivalent to the condition
\begin{equation*}
U(g)Q(f)U(g)^{-1}=Q(L_gf) 
\end{equation*} 
for all $g\in G$ and $f\in C(X,[0,1])$ where $Q$ is quantization of events, and $L_g$ is a left translation defined by $(L_gf)(x)=f(g^{-1}x)$ for all $f\in C(X,[0,1])$, $x\in X$ and $g\in G$.
So a generalized system of imprimitivity can be represented in a commuting diagram involving quantization of events:\\
\begin{tikzpicture}[auto]
\node (a) at (0, 1.5) {$C(X,[0,1])$}; \node (x) at (3, 1.5) {$Ef(\mathcal{H})$};
\node (b) at (0, 0) {$C(X,[0,1])$};   \node (y) at (3, 0) {$Ef(\mathcal{H})$};
\draw[->] (a) to node {$Q$} (x);
\draw[->] (x) to node {$U(g)(-)U(g)^{-1}$} (y);
\draw[->] (a) to node[swap] {$L_g$} (b);
\draw[->] (b) to node[swap] {$Q$}(y);
\end{tikzpicture}

By taking a dual of this diagram, we can see what a dual of a generalized system of imprimitivity is.
Taking a dual of the left part of the diagram yields a dual of $L_g$:\\
\begin{tikzpicture}[auto]
\node (a) at (0, 1.5) {$C(X,[0,1])$}; \node (x) at (3, 1.5) {$[0,1]$};
\node (b) at (0, 0) {$C(X,[0,1])$};   \node (y) at (3, 0) {$[0,1]$};
\node (s) at (4.5,0) {$PR(X)$};   \node (t) at (4.5,1.5) {$PR(X)$};
\draw[->] (a) to node {$\int_X(-)d\mu_g$} (x);
\draw[->] (x) to node {$id$} (y);
\draw[->] (a) to node[swap] {$L_g$} (b);
\draw[->] (b) to node[swap] {$\int_X(-)d\mu$} (y);
\draw[->] (s) to node[swap] {$(-)_g$} (t);
\end{tikzpicture}
\\where $\mu_g(\Delta)=\mu(g\Delta)$ for all $g\in G$ and all Borel sets $\Delta\subset X$.
Also, taking a dual of the right part of the diagram yields a dual of $U(g)(-)U(g)^{-1}$:\\
\begin{tikzpicture}[auto]
\node (a) at (0, 1.5) {$Ef(\mathcal{H})$}; \node (x) at (3, 1.5) {$[0,1]$};
\node (b) at (0, 0) {$Ef(\mathcal{H})$};   \node (y) at (3, 0) {$[0,1]$};
\node (s) at (4.5,0) {$DM(\mathcal{H})$};   \node (t) at (4.5,1.5) {$DM(\mathcal{H})$};
\draw[->] (a) to node {$\mathrm{Tr}[U(g)^{-1}\rho U(g)(-)]$} (x);
\draw[->] (x) to node {$id$} (y);
\draw[->] (a) to node[swap] {$U(g)(-)U(g)^{-1}$} (b);
\draw[->] (b) to node[swap] {$\mathrm{Tr}[\rho(-)]$} (y);
\draw[->] (s) to node[swap] {$U(g)^{-1}(-)U(g)$} (t);
\end{tikzpicture}
\\Summing up these, we get the dual of a generalized system of imprimitivity:\\
\begin{tikzpicture}[auto]
\node (a) at (0, 1.5) {$DM(\mathcal{H})$}; \node (x) at (3, 1.5) {$PR(X)$};
\node (b) at (0, 0) {$DM(\mathcal{H})$};   \node (y) at (3, 0) {$PR(X)$};
\draw[->] (a) to node {$M$} (x);
\draw[->] (x) to node {$(-)_g$} (y);
\draw[->] (a) to node[swap] {$U(g)^{-1}(-)U(g)$} (b);
\draw[->] (b) to node[swap] {$M$}(y);
\end{tikzpicture}
\\where $M$ is mesurement of states.
This commuting diagram gives the identity
\begin{equation*}
\mathrm{Tr}[\rho E(g\Delta)]=\mathrm{Tr}[\rho U(g)E(\Delta)U(g)^{-1}]
\end{equation*}
for all $\rho\in DM(\mathcal{H})$, $g\in G$ and all Borel sets $\Delta\subset X$.
This condition states that measurement $E$ is covariant. 

So, we have obtained a covariant measurement as a dual of a generalized system of imprimitivity.
For a detailed explanation of a covariant measurement, see \cite{AH}.

Conversely, we can also induce a generalized system of imprimitivity from a covariant measurement.
From this discussion, we conclude that a generalized system of imprimitivity and a covariant measurement are in a dual relation

Second example is a measurement model.
A measurement model describes a successive procedure that a system and a probe become coupled, they evolves, and a measurement is carried out on the probe.
For a detailed explanation, see \cite{PBB}.

A measurement model is quadruple $(\mathcal{K}, \rho_0, \Lambda, F)$ where $\mathcal{K}$ is a Hilbert space of the probe, $\rho_0$ is an initial state of the probe,
 $\Lambda$ is a completely positive trace preserving (CPTP) map from $B(\mathcal{H}\otimes\mathcal{K}$) to $B(\mathcal{H}\otimes\mathcal{K}$) and $F$ is a POVM on $\mathcal{K}$.
 If the following condition holds for a POVM $E$ on $\mathcal{H}$, it is called a measurement model for $E$:
 \begin{equation*}
  \mathrm{Tr}_\mathcal{H}[\rho E]=\mathrm{Tr}[\Lambda(\rho\otimes\rho_0)(I\otimes F)]
 \end{equation*}
 for all $\rho\in DM(\mathcal{H})$. 
A measurement model for $E$ can be represented in a commuting diagram involving measurement of events:\\
\begin{tikzpicture}[auto]
\node (s) at (0, 3) {$DM(\mathcal{H})$};
\node (a) at (0, 1.5) {$DM(\mathcal{H}\otimes\mathcal{K})$}; \node (x) at (4 ,1.5) {$PR(X)$};
\node (b) at (0, 0) {$DM(\mathcal{H}\otimes\mathcal{K})$};   
\draw[->] (s) to node[swap] {$(-)\otimes\rho_0$} (a);
\draw[->] (a) to node[swap] {$\Lambda$} (b);
\draw[->] (s) to node {$\mathrm{Tr}_\mathcal{H}[(-)E]$} (x);
\draw[->] (b) to node[swap] {$\mathrm{Tr}[(-)I\otimes F]$} (x);
\end{tikzpicture}

By taking a dual of this diagram, we can see what a dual of a measurement model for $E$ is.
Taking a dual of the upper component of the left of diagram yields a dual of $(-)\otimes\rho_0$:\\
\begin{tikzpicture}[auto]
\node (a) at (0, 1.5) {$DM(\mathcal{H})$}; \node (x) at (3, 1.5) {$[0,1]$};
\node (b) at (0, 0) {$DM(\mathcal{H}\otimes\mathcal{K})$};   \node (y) at (3, 0) {$[0,1]$};
\node (s) at (5,0) {$Ef(\mathcal{H}\otimes\mathcal{K})$};   \node (t) at (5,1.5) {$Ef(\mathcal{H})$};
\draw[->] (a) to node {$\mathrm{Tr}_\mathcal{H}[(-)\mathrm{Tr}_\mathcal{K}[\rho_0 A]]$} (x);
\draw[->] (x) to node {$id$} (y);
\draw[->] (a) to node[swap] {$(-)\otimes\rho_0$}(b);
\draw[->] (b) to node[swap] {$\mathrm{Tr}[(-)A]$} (y);
\draw[->] (s) to node[swap]  {$\mathrm{Tr}_\mathcal{K}[\rho_0 (-)]$} (t);
\end{tikzpicture}
\\Also, taking a dual of the lower component of the left of diagram yields a dual of $\Lambda$:\\
\begin{tikzpicture}[auto]
\node (a) at (0, 1.5) {$DM(\mathcal{H}\otimes\mathcal{K})$}; \node (x) at (3, 1.5) {$[0,1]$};
\node (b) at (0, 0) {$DM(\mathcal{H}\otimes\mathcal{K})$};   \node (y) at (3, 0) {$[0,1]$};
\node (s) at (5,0) {$Ef(\mathcal{H}\otimes\mathcal{K})$};   \node (t) at (5,1.5) {$Ef(\mathcal{H}\otimes\mathcal{K})$};
\draw[->] (a) to node {$\mathrm{Tr}[(-)\Lambda_*(A)]$} (x);
\draw[->] (x) to node {$id$} (y);
\draw[->] (a) to node[swap] {$\Lambda $}(b);
\draw[->] (b) to node[swap] {$\mathrm{Tr}[(-)A]$} (y);
\draw[->] (s) to node[swap]  {$\Lambda_*$} (t);
\end{tikzpicture}
\\where $\Lambda_*$ is a dual of the CPTP map $\Lambda$ defined by $\mathrm{Tr}[\Lambda(T)B]=\mathrm{Tr}[T\Lambda_*(B)]$ for all $T\in T(\mathcal{H})$ and $B\in B(\mathcal{H})$.
Summing up these, we get a dual of a measurement model for $E$:\\
\begin{tikzpicture}[auto]
\node (s) at (4, 3) {$Ef(\mathcal{H}\otimes\mathcal{K})$};
\node (a) at (0, 1.5) {$C(X,[0,1])$}; \node (x) at (4, 1.5) {$Ef(\mathcal{H}\otimes\mathcal{K})$};
 \node (y) at (4, 0) {$Ef(\mathcal{H})$};
\draw[->] (a) to node {$\int_X(-)d(I\otimes F)$} (s);
\draw[->] (s) to node {$\Lambda_*$} (x);
\draw[->] (x) to node {$\mathrm{Tr}_\mathcal{K}[\rho_0(-)]$} (y);
\draw[->] (a) to node[swap] {$\int_X(-)dE$}(y);
\end{tikzpicture}
\\This commuting diagram gives the identity
\begin{equation*}
\int_X fdE=\mathrm{Tr}_\mathcal{K}[\rho_0\Lambda_*(\int_X fd(I\otimes F))]
\end{equation*}
for all $f\in C(X,[0,1])$.

Since $C(X,[0,1])$ can be extended to $C(X)$, this is a quantization scheme involving the probe and time evolution of the whole system.
We call it a dual measurement model quantization for $E$.
This is quantization such that the successive procedure of quantization to the whole system,a  time evolution, and a discard of the probe by the initial state is equivalent to quantization to the system. 

Of course, it is easy to check the fact that a dual of a dual  measurement model quantization for $E$ is a measurement model for $E$.

\begin{acknowledgments}
We wish to acknowledge the support of N.P. Landsman for the discussion, the support of the research and encouragement.
We also wish to acknowledge the support of  K. Matsuhisa and J. Lee for the discussion and I. Tsutsui for the supervision.
We acknowledge the financial support of Advanced Leading Photon Science program and Graduate Research Abroad Science Program.
\end{acknowledgments}

\nocite{*}
\bibliography{dbqm}

\end{document}